\documentclass[conference,onecolumn,11pt]{IEEEtran}
\usepackage{algorithm,algorithmic}
\usepackage{url}

\usepackage{graphicx}
\usepackage{color}
\usepackage{chapterbib}
\usepackage{pdfsync}
\usepackage[hang]{subfigure}

\usepackage{times}
\usepackage{latexsym}
\usepackage{bbm,dsfont}
\usepackage{amsmath}
\usepackage{amssymb}
\usepackage{latexsym}
\usepackage{amsthm}
\usepackage{amsfonts}
\usepackage{epsfig}
\usepackage{latexsym}
\usepackage{cite}
\usepackage{graphicx}
\usepackage{amsmath}
\usepackage{amsthm}
\renewcommand{\figurename}{Fig.}

\newcommand{\bea}{\begin{eqnarray}}
\newcommand{\eea}{\end{eqnarray}}

\usepackage{epsf}
\usepackage{amssymb}
\usepackage{graphicx}
\usepackage{amsmath}
\usepackage[english]{babel}
\usepackage{mathrsfs}

\usepackage{setspace}

\usepackage{ifthen}
\usepackage{color}
\usepackage{amssymb}

\usepackage[newenum]{paralist}

\newcommand{\IsTR}{\ifthenelse{1<2}}

\IsTR{}{ \ninept }

\DeclareMathOperator*{\argmax}{arg\max}  

\newcommand{\set}[1]{\mathbb{#1}}
\newtheorem{theorem}{Theorem}

\DeclareMathSizes{10}{8}{7}{5}
\DeclareMathSizes{11}{11}{7.7}{5.5}
\DeclareMathSizes{12}{12}{8.4}{6}
\DeclareMathSizes{13}{13}{9.1}{6.5}
\DeclareMathSizes{14}{14}{9.8}{7}

\begin{document}
\title{An Adaptive Multi-channel P2P Video-on-Demand System using Plug-and-Play Helpers}
\author{
Hao Zhang{$^{\dag}$}, Minghua Chen{$^{\S}$}, Abhay Parekh{$^{\dag}$} and Kannan Ramchandran{$^{\dag}$}\\

\begin{tabular}{c@{}}
       {$^{\dag}$}Department of EECS\\
       University of California at Berkeley, Berkeley, CA\\
       \{zhanghao, parekh, kannanr\}@eecs.berkeley.edu
\end{tabular} \hskip 0.15in \begin{tabular}{c@{}}
       {$^{\S}$}Department of Information Engineering\\
       The Chinese University of Hong Kong, Hong Kong\\
       minghua@ie.cuhk.edu.hk
\end{tabular}
}
\maketitle


\begin{abstract}
We present a multi-channel P2P Video-on-Demand (VoD) system using ``plug-and-play" helpers. Helpers are heterogenous ``micro-servers" with limited storage, bandwidth and number of users they can serve simultaneously. Our proposed system has the following salient features: (1) it minimizes the server load; (2) it is distributed, and requires little or no maintenance overhead and which can easily adapt to system dynamics; and (3) it is adaptable to varying supply and demand patterns across multiple video channels irrespective of video popularity. Our proposed solution jointly optimizes over helper-user topology, video storage allocation and bandwidth allocation. The combinatorial nature of the problem and the system demand for distributed algorithms makes the problem uniquely challenging. By utilizing Lagrangian decomposition and Markov chain approximation based arguments, we address this challenge by designing two distributed algorithms running in tandem: a primal-dual storage and bandwidth allocation algorithm and a ``soft-worst-neighbor-choking" topology-building algorithm. Our scheme provably converges to a near-optimal solution, and is easy to implement in practice. Simulation results validate that the proposed scheme achieves minimum sever load under highly heterogeneous combinations of supply and demand patterns, and is robust to system dynamics of user/helper churn, user/helper asynchrony, and random delays in the network.
\end{abstract}

\section{Introduction}
Our paper is motivated by the following characteristics of online video traffic:
\begin{itemize}
\item The amount of video traffic is growing exponentially~\cite{cisco,huang2007civ}: YouTube estimates that $24$ hours of video are uploaded to its site every minute; thousands of films and TV shows are available for streamed viewing from sites such as Netflix, Amazon and iTunes. Cisco projects that video will account for $60\%$ of the Internet traffic by $2013$.
\item The demand for video titles is time-varying and heavy-tailed~\cite{boufkhad2008achievable}. The need for on-demand video delivery, where users can ``channel surf'' by switching across the menu of thousands of available videos, is rapidly upsurging.
\end{itemize}
A fundamental challenge in supporting such a large and diverse on-demand infrastructure is the degree to which this infrastructure can be distributed and maintained at low cost. It seems clear that architectures biased towards centralized distribution are not likely to be scalable as video adoption grows. YouTube, for example, pays millions of dollars per month only on bandwidth costs~\cite{huang2007civ}. Peer-to-Peer (P2P) systems save cost and adjust load automatically, yet they do not provide acceptable quality of service (QoS) for the viewers of all but the most popular videos.

To match the heavy-tailed demand patterns at low maintenance cost, researchers and engineers have introduced the concept of helpers and explored the design of helper-assisted P2P VoD systems. Helper nodes are ``micro-servers," which, in the system scale, have only limited individual resources of storage and bandwidth for the video streaming service. In the PPStream system for example~\cite{ppstream}, each peer dedicates about $1$ GB of its local storage to cache previously watched videos and helps serve users, therefore reducing the load on the central server. The concept of helpers has also been used in other P2P streaming applications including Xunlei~\cite{xunlei} and PPLive~\cite{pplive}. In this paper, we envision an ecosystem in which a variety of devices, including set top boxes, inexpensive PCs and small servers such as CDNs, are incentivized to participate in an economy of helpers.

Minimizing the server load in a helper-assisted P2P VoD system is a challenging problem under practical constraints. First, the extreme range of videos and the sheer amount of content makes it impractical to store and serve every video on any individual helper. Given a number of distributed helper nodes with limited resources of storage and bandwidth in the system scale, it is important to answer the questions of ``what fraction of what videos should be stored on each helper to optimize overall system demand patterns" and ``how much bandwidth should helpers allocate to each of their requesting users?" Second, due to practical connection overheads, there are limits on how many users each helper can simultaneously connect to and vice versa. Therefore, an important question is ``how should we build optimal helper-user overlay topology?" Since the number of helper-user topology configurations is exponential in the number of nodes, topology building is a challenging combinatorial problem. Third, video popularity is time-varying, and helper/user may randomly join and leave the system, making it difficult to keep track of the supply and demand patterns in real time. It is desirable that the system can adapt to these fluctuations with minimum overhead.

In this paper, we present a helper-assisted P2P VoD system to solve these challenges. Our system has the following distinguishing attributes:
\begin{itemize}
\item It minimizes the server load. To do this, we jointly optimize over helper's video storage and bandwidth allocations, and helper-user topology. By utilizing Lagrangian decomposition and Markov chain approximation based arguments, we design two distributed algorithms running them in tandem: a primal-dual resource allocation algorithm and a ``soft-worst-neighbor-choking'' topology-building algorithm.
\item The proposed algorithm is fully distributed and easy to implement. Peers can dynamically optimize system resources by taking actions based on only local information yet being able to achieve global optimality with provable convergence. Helpers are ``plug-and-play'', i.e., a newly deployed helper will automatically connect to a set of interested users and load balance its storage and bandwidth according to the up-to-date supply and demand patterns with minimum maintenance requirement.
\item Thanks to the simple distributed solution, our proposed system is adaptable to varying supply and demand patterns across multiple video channels irrespective of video popularity. Our system achieves minimum sever load under highly heterogeneous combinations of supply and demand patterns, and is robust to system dynamics of user/helper churn, user/helper asynchrony, and random delays in the network.
\end{itemize}
Simulation results validate the feasibility and effectiveness of the proposed algorithms and offer new insights into building practical P2P VoD systems.

\section{Related Work}
Distributed bandwidth allocation for P2P streaming was studied in several works~\cite{wu2008meeting,wang08acmmm}. Wu and Li proposed a rate allocation scheme for a single-video P2P live streaming application without helpers and showed its convergence and optimality~\cite{wu2008meeting}. Wang et al. proposed a distributed solution to minimize the weighted sum of server load and non-ISP-friendly traffic~\cite{wang08acmmm} under a single video scenario. The problem of being able to switch streams (channels) has been studied in the live P2P streaming case by Wu et al.~\cite{wu2009queuing} who propose the concept of view-coupling~\cite{wang2008epp,zhang2009icip}. When there is storage constraint, Huang and et al.~\cite{huang2008challenges} suggested using the proportional replication strategy, i.e., replicating video storage in the system proportionally to their demand. However, this ignores the available system bandwidth for the videos, and can result in poor performance for videos with low demand as was observed by Wang and Lin~\cite{wang2009}. Optimal multi-channel on-demand solutions with both bandwidth and storage constraints are challenging primarily because it is difficult to effectively aggregate instantaneous user demand and keep track of available system resources distributively.

The concept of helpers has been addressed by several authors~\cite{wang2008epp,zhang2009icip,he2009improving}. Wang et al. studied a single-video helper-assisted P2P live streaming scenario and proposed that each helper downloads only one coded packet of the segment that is currently being streamed. Simulation results showed that the proposed system can achieve significantly improved streaming bitrate without incurring additional server load. Zhang et al.~\cite{he2009improving} and He et al.~\cite{zhang2009icip} individually proposed similar concepts of using helpers in a single-video P2P VoD application to boost system performance. He et al. also proposed a distributed algorithm to allocate bandwidth among helpers and users assuming a given helper assignment and a fixed helper-user connection topology. However, all of these works focus on a single video scenario without considering helpers' storage constraints imposed by the sheer amount of aggregate video volume when tackling the common scenario having an arbitrarily large number of videos. Furthermore, fixed overlay topology was assumed and was not optimized for the overall system performance. BitTorrent~\cite{cohen2003incentives} uses a ``worst-neighbor-choking" algorithm to update overlay topology, in which users periodically connects to a new and randomly chosen neighbor and chokes the worst performing neighbor. In contrast, we propose a ``soft-worst-neighbor-choking" algorithm and prove its optimality.

There are a number of works on practical P2P VoD system design. Annapureddy and et al. proposed a P2P VoD system called Redcarpet~\cite{annapureddy2006providing}. The authors proposed an efficient video block dissemination algorithm in a mesh-based P2P system, and showed that pre-fetching and network coding techniques can greatly improve system performance. Simulation results showed their system can achieve small start-up time and smooth video playback. Huang et al. studied~\cite{huang2008challenges} the challenges and the architectural design issues of a large-scale P2P-VoD system based on the experiences of a real system deployed by PPLive. Such challenges include coordinating content storage distribution, content discovery, and peer scheduling. There are also a number measurement studies of practical P2P VoD systems~\cite{hei2007measurement,huang2007civ}.

Our work contributes to the VoD literature in several respects: First, we take practical considerations of bounded user/helper connections and propose a distributed algorithm that optimizes the overlay topology building. Second, we target the problem of multi-channel helper-assisted P2P VoD system with both storage and bandwidth constraints and jointly optimize their allocation. Third, our distributed solution is provided with provable analytics that allows for videos to be served efficiently and with effective response time irrespective of video popularity. Our system is ``plug-and-play", and can easily handle user/helper dynamics and video demand pattern changes and it does so with minimal central co-ordination.

\section{Problem Setup}
\label{sec:problem}
We first formulate the problem of minimizing the server load in a static helper-assisted VoD system, and then design distributed solutions that allows the system to be adaptive to varying demand patterns in dynamic situations with minimum maintenance overhead. Table~\ref{tab:notation} lists the relevant notations.

\subsection{Problem Overview and Assumptions}
\label{sec:assump}
Consider a VoD system where $M$ videos are served by a dedicated central server and a group of helpers. The dedicated central server fills in any required system deficit to guarantee the real-time streaming requirements of all the users. Each video $m$ has constant streaming rate $r_m$, duration $l_m$ and size $V_m=r_ml_m$, $m=1,2,\ldots,M$. There are $I_m$ users in each video session $m$, and every user $i_m$ watches only one video at a time. Consider $J$ helpers in the system, and each helper $j$ has storage capacity $S_j$ and upload bandwidth capacity $B_j$, $j=1,2,\ldots,J$.

We assume that the download bandwidths do not represent a system bottleneck, as is true in typical Internet access scenarios today. We also assume decoupled roles of users and helpers, i.e., users can only request service and helpers can only provide service. In the case that users can also help redistribute the video content, one can still conceptually separate their roles as service-request and service-offering entities, as proposed by Zhang et al.~\cite{zhang2009icip} and Wu et al.~\cite{wu2009queuing}. We will refer to both users and helpers as ``peers''.

Due to practical limits in connection overhead, we consider that each user $i_{m}$ (helper $j$) can \emph{simultaneously} connect to at most $N_{i_{m}}^{\max}$ ($N_{j}^{\max}$) neighbor nodes from its candidate neighbor set $\set{N}_{i_{m}}$ ($\set{N}_{j}$). Denote by set $\set{C}$ the entire possible helper-user topology configurations, where each configuration $c\in\set{C}$ is a set of links which connect all the users and helpers and which satisfy the bounded-neighbors constraints. Denote by $\set{N}_{i_{m}}^{c}$ ( $\set{N}_{j,m}^{c}$) the set of active helpers (users) that user $i_{m}$ (helper $j$) connects to under configuration $c$. Helper $j$'s upload rate to user $i_{m}\in\set{N}_{j,m}^{c}$ under $c$ is denoted by $x_{ji_{m}}^{c}$. For simplicity, we will
drop the superscript $c$ in $x_{ji_{m}}^{c}$ and use $x_{ji_{m}}$ instead, but it should be easy to clarify based on the context.
\begin{table}[t!]
{\centering \caption{Key Notations}
\label{tab:notation} \vspace{-0.1in}
 \tabcolsep 5.8pt {\small }\begin{tabular}{l|l}
\hline
\textbf{\small Notation}{\small{} } & \textbf{\small Definition}\tabularnewline
\hline
{\small $M$ } & {\small total number of videos }\tabularnewline
{\small $I_{m}$ } & {\small total number of users watching video $m$}\tabularnewline
{\small $J$ } & {\small total number of helpers }\tabularnewline
{\small $r_{m},l_{m},V_{m}$ } & {\small video $m$'s streaming rate, duration and size }\tabularnewline
{\small $\set{C}$ } & {\small set of feasible overlay configurations}\tabularnewline
{\small $\set{N}_{i_{m}}$ } & {\small user $i_{m}$'s helper neighborhood}\tabularnewline
{\small $\set{N}_{j}$ } & {\small helper $j$'s user neighborhood }\tabularnewline
{\small $\set{N}_{i_{m}}^{c}$ } & {\small set of helpers connected to user $i_{m}$ under $c$}\tabularnewline
{\small $\set{N}_{j,m}^{c}$ } & {\small set of users connected to helper $j$ under $c$}\tabularnewline
{\small $B_{j},S_{j}$ } & {\small upload, storage capacity of helper $j$}\tabularnewline
{\small $x_{ji_{m}}$ } & {\small upload rate from helper $j$ to user $i_{m}$}\tabularnewline
{\small $f_{jm}$ } & {\small fraction of video $m$ stored by helper $j$, in $[0,1]$ }\tabularnewline
{\small $k_{ji_{m}}$ } & {\small helper $j$'s availability price to user $i_{m}$}\tabularnewline
{\small $\lambda_{j},\mu_{j}$ } & {\small bandwidth, storage prices of helper $j$}\tabularnewline
\hline
\end{tabular}
}
Note: we use bold-type to denote vectors.
\end{table}

We break video $m$ into segments, each having $k$ packets of equal size, e.g., 1KB. A helper increases (decreases) its stored portion of video $m$, by downloading (offloading) at the unit of one packet per segment for all video $m$'s segments. Each packet that the helper stores for each segment is coded using a random linear combination of the $k$ original packets of that segment. In this way, any coded packets from the helpers are equally useful to users in need of the corresponding segments. Consequently, a helper with $f_{jm}$ fraction of the video $m$ can supply at the rate of $f_{jm}r_{m}$ to each user in video session $m$, regardless of users' playback times and what they receive from other helpers. The above coding and storage arrangement simplifies the system design as well as the problem formulation.

\subsection{Problem Formulation}
Minimizing the server load is equivalent to maximizing the sum of helpers' effective contribution to all the users: $\max_{c,\mathbf{f},\mathbf{x}}{\sum_{m=1}^{M}\sum_{i_m=1}^{I_m}\min(\sum_{j\in \set{N}^c_{i_m}}x_{ji_m},r_m)}$.
Here we implicitly assume that helper's download is a transient cost and is negligible compared to their contribution. Incorporating the constraints, we arrive at the following optimization problem:
\begin{eqnarray}
\max_{c,\mathbf{f},\mathbf{x}} &  & \sum_{m=1}^{M}\sum_{i_{m}=1}^{I_{m}}\min(\sum_{j\in\set{N}_{i_{m}}^{c}}x_{ji_{m}},r_{m})\label{equa:convex}\\
\mbox{s.t.} &  & x_{ji_{m}}\leq f_{jm}r_{m},\;\forall j,m,i_{m}\in\set{N}_{j,m}^{c}\label{equa:cons1}\\
 &  & \sum_{m=1}^{M}\sum_{i_{m}\in\set{N}_{j,m}^{c}}x_{ji_{m}}\leq B_{j},\;\forall j\label{equa:cons2}\\
 &  & \sum_{m=1}^{M}f_{jm}V_{m}\leq S_{j},\;\forall j\label{equa:cons3}\\
 &  & 0\leq f_{jm}\leq1,\;\forall j,m\label{equa:cons4}\\
 &  & c\in\set{C}\label{equa:cons5}\end{eqnarray}
Constraints~$(\ref{equa:cons1})$ are such that each helper's total upload to neighboring user $i_{m}$ who is viewing video $m$ cannot exceed its available service rate for video $m$. Constraints~$(\ref{equa:cons2})$, $(\ref{equa:cons3})$ and $(\ref{equa:cons4})$ are feasibility constraints on bandwidth and storage. $(\ref{equa:cons5})$ are combinatorial,
representing the bounded-neighbor helper-user topology constraints. The above problem is a joint storage, bandwidth, and helper-user topology optimization problem, and is challenging to solve even in a centralized manner due to its combinatorial nature. In the next section, we design distributed algorithms and prove its convergence to a near-optimal solution.

\section{Distributed Solution}
\label{sec:distributed}
Constraints~(\ref{equa:cons1})$\sim$(\ref{equa:cons4}) are independent of constraint~(\ref{equa:cons5}), which allows us to decompose it into a resource allocation problem and a topology building problem and solve them in tandem. We present our solutions to each problem in the following subsections.
\subsection{Storage and Bandwidth Allocation}
Fixing topology $c$, and assigning Lagrangian variables $k_{ji_m}$ to constraints~$(\ref{equa:cons1})$, $\lambda_j$ to constraints~$(\ref{equa:cons2})$, and $\mu_j$ to constraints~$(\ref{equa:cons3})$, we obtain the following partial Lagrangian:
\begin{eqnarray}\label{equa:fixc}
    \min_{\mathbf{k},\mathbf{\lambda},\mathbf{\mu}}\max_{\mathbf{0}\leq\mathbf{f}\leq \mathbf{1},\mathbf{x}}\sum_{m,i_m}\min(\sum_{j\in \set{N}^c_{i_m}}x_{ji_m},r_m)+\sum_{j,m,i_m}\lambda_j(B_j - x_{ji_m})\nonumber\\
    +\sum_{j,m,i_m}k_{ji_m}(f_{jm}r_m-x_{ji_m})+\sum_{j}\mu_j(S_j-\sum_{m}f_{jm}V_m)
\end{eqnarray}
whose optimal solution is denoted by $U(c)$. For simplicity, we abbreviated ${\displaystyle\sum_{m=1}^{M}\sum_{i_m=1}^{I_m}}$ as ${\displaystyle\sum_{m,i_m}}$, and ${\displaystyle\sum_{j=1}^J\sum_{m=1}^{M}\sum_{i_m\in\set{N}^c_{j,m}}}$ as ${\displaystyle\sum_{j,m,i_m}}$. By rearranging the terms, the above problem can be solved successively in the primal and dual variables. Specifically, given $\mathbf{k},\mathbf{\lambda}$ and $\mathbf{\mu}$, we have the subproblem in $\mathbf{x}$ and $\mathbf{f}$:
\begin{eqnarray}\label{equa:primal}
    \max_{\mathbf{x}}\sum_{m,i_m}(\min(\sum_{j\in \set{N}^c_{i_m}}x_{ji_m},r_m)-\sum_{j\in \set{N}^c_{i_m}}(\lambda_j+k_{ji_m})x_{ji_m})\nonumber\\
    + \max_{\mathbf{0}\leq\mathbf{f}\leq \mathbf{1}}\sum_{j,m}(r_m\sum_{i_m\in \set{N}^c_{j,m}}k_{ji_m}-\mu_jV_m)f_{jm}
\end{eqnarray}
The unique structure of the partial Lagrangian allows us to use a simple primal-dual algorithm~\cite{chen2008ump} as its solution, which we state in the following theorem.
\begin{theorem}
The following resource allocation algorithm converges to the optimal solution to problem~(\ref{equa:fixc}):
\begin{eqnarray}\label{equa:solver}
\left\{
  \begin{array}{ll}
    \dot{x}_{ji_m} = \alpha\left(g_{x_{ji_m}}-(\lambda_j+k_{ji_m})\right)_{x_{ji_m}}^{[0,+\infty)},\;\forall j,m,i_m\in\set{N}^c_{j,m}\\
    \dot{f}_{jm} = \beta(\sum_{i_m\in \set{N}^c_{j,m}}k_{ji_m}-l_m\mu_j)_{f_{ji_m}}^{[0,1]},\;\forall j,m\\
    \dot{\lambda}_j = \gamma(\sum_{m=1}^{M}\sum_{i_m\in \set{N}^c_{j,m}}x_{ji_m} - B_j)_{\lambda_j}^{[0,+\infty)},\;\forall j\\
    \dot{\mu}_j = \delta(\sum_{m=1}^{M}f_{ji_m}V_m - S_j)_{\mu_j}^{[0,+\infty)},\;\forall j\\
    \dot{k}_{ji_m} = \varepsilon(x_{ji_m}-f_{jm}r_m)_{k_{ji_m}}^{[0,+\infty)},\;\forall j,m,i_m\in\set{N}^c_{j,m}.
  \end{array}
\right.
\end{eqnarray}
where $
h_y^{[a,b]} = \left\{
                       \begin{array}{ll}
                         \min(0,h), & \hbox{$y\geq b$;} \\
                         h, & \hbox{$a<y<b$;} \\
                         \max(0,h), & \hbox{$y\leq a$,}
                       \end{array}
                     \right.$
$g_{x_{ji_m}}$ is the partial derivative of the function $g$ with respect to $x_{ji_m}$ and $g=\min(\sum_{j\in \set{N}^c_{i_m}}x_{ji_m},r_m)$, and $\alpha,\beta,\gamma,\delta,\varepsilon$ are small step sizes.
\end{theorem}
\begin{proof}
At optimal, the following KKT conditions of the Lagrangian in~(\ref{equa:fixc}) should also hold:
\begin{eqnarray}
\left\{
  \begin{array}{ll}
    (g_{x^{\ast}_{ji_m}}-(\lambda^{\ast}_j+k^{\ast}_{ji_m}))_{x^{\ast}_{j_im}}^{[0,+\infty)} = 0 \\
    (\sum_{i_m\in \set{N}_j^m}k^{\ast}_{ji_m}-l_m\mu^{\ast}_j)_{f^{\ast}_{ji_m}}^{[0,1]} = 0  \\
    \lambda^{\ast}_j(\sum_{m=1}^{M}\sum_{i_m\in \set{N}_j^{m}}x^{\ast}_{ji_m} - B_j)=0\\
    \mu^{\ast}_j(\sum_{m=1}^{M}f^{\ast}_{ji_m}V_m - S_j)=0\\
    k^{\ast}_{ji_m}(x^{\ast}_{ji_m}-f^{\ast}_{jm}r_m)=0
  \end{array}
  \nonumber
\right.
\end{eqnarray}
where $\mathbf{x}^{\ast}$ and $\mathbf{f}^{\ast}$ are the primal optimal and $\mathbf{\lambda}^{\ast},\mathbf{\mu}^{\ast}$ and $\mathbf{k}^{\ast}$ are the dual optimal. Strictly speaking, $g$ is not differentiable at $\sum_{j\in \set{N}^c_{i_m}}x_{ji_m}=r_m$, where we simply let $g_{x_{ji_m}}=0$ in practice without affecting the performance. Denote by $\mathbf{y}=(\mathbf{x},\mathbf{f},\mathbf{\lambda},\mathbf{\mu},\mathbf{k})$ and by $\mathbf{y^{\ast}}=(\mathbf{x^{\ast}},\mathbf{f^{\ast}},\mathbf{\lambda^{\ast}},\mathbf{\mu^{\ast}},\mathbf{k^{\ast}})$. To prove that $\mathbf{y}\rightarrow \mathbf{y^{\ast}}$, we propose the following generalized energy function:
\begin{eqnarray}
    V(\mathbf{y}) & = & \frac{1}{2\alpha}\|\mathbf{x}-\mathbf{x^{\ast}}\|^2 + \frac{1}{2\beta}\|\mathbf{f}-\mathbf{f^{\ast}}\|^2 + \frac{1}{2\gamma}\|\mathbf{\lambda}-\mathbf{\lambda^{\ast}}\|^2 \nonumber\\
    &  & +\frac{1}{2\varepsilon}\|\mathbf{k}-\mathbf{k^{\ast}}\|^2 + \frac{1}{2\delta}\|\mathbf{\mu}-\mathbf{\mu^{\ast}}\|^2
    \nonumber
\end{eqnarray}
and show that (a) $V(\mathbf{y}) > 0\;\forall \mathbf{y}\neq \mathbf{y^{\ast}}$ and $V(\mathbf{y^{\ast}})=0$; (b) $\dot{V}(\mathbf{y}) \leq 0\;\forall \mathbf{y}$ and $\dot{V}(\mathbf{y^{\ast}})=0$.

(a) is obvious since $V(\mathbf{y})$ is summation of quadratic terms. To show (b), we derive $\dot{V}(\mathbf{y})$:
\begin{eqnarray}
    \dot{V}(\mathbf{y}) & = & \sum(x_{ji_m}-x^{\ast}_{ji_m})(g_{x_{ji_m}}-(\lambda_j+k_{ji_m}))_{x_{j_im}}^{[0,+\infty)}\nonumber\\
      &  & + \sum(f_{jm}-f^{\ast}_{jm})(\sum k_{ji_m}-l_m\mu_j)_{f_{ji_m}}^{[0,1]}\nonumber\\
      &  & + \sum(\lambda_{j}-\lambda^{\ast}_{j})(\sum x_{ji_m} - B_j)_{\lambda_j}^{[0,+\infty)}\nonumber\\
      &  & + \sum(k_{ji_m}-k^{\ast}_{ji_m})(x_{ji_m}-f_{jm}r_m)_{k_{ji_m}}^{[0,+\infty)}\nonumber\\
      &  & + \sum(\mu_{j}-\mu^{\ast}_{j})(\sum f_{ji_m}V_m- S_j )_{\mu_j}^{[0,+\infty)}\nonumber
\end{eqnarray}
by applying partial derivatives and plugging in the dynamic system equations. For simplicity, we have omitted the sets over which the terms are summed up. It is easy to see that $\dot{V}(\mathbf{y^{\ast}})=0$. Now we can upper-bound $\dot{V}(\mathbf{y})$ as follows:
\begin{eqnarray}
    \dot{V}(\mathbf{y}) & \leq & \sum(x_{ji_m}-x^{\ast}_{ji_m})(g_{x_{ji_m}}-(\lambda_j+k_{ji_m}))\nonumber\\
      &  & + \sum(f_{jm}-f^{\ast}_{jm})(\sum k_{ji_m}-l_m\mu_j)\nonumber\\
      &  & + \sum(\lambda_{j}-\lambda^{\ast}_{j})(\sum x_{ji_m} - B_j)\nonumber\\
      &  & + \sum(k_{ji_m}-k^{\ast}_{ji_m})(x_{ji_m}-f_{jm}r_m)\nonumber\\
      &  & + \sum(\mu_{j}-\mu^{\ast}_{j})(\sum f_{ji_m}V_m- S_j)\nonumber\\
      & = & \sum(x_{ji_m}-x^{\ast}_{ji_m})(g_{x_{ji_m}}-g_{x^{\ast}_{ji_m}})\nonumber\\
      &  & + \sum(x_{ji_m}-x^{\ast}_{ji_m})(g_{x^{\ast}_{ji_m}}-(\lambda^{\ast}_j+k^{\ast}_{ji_m}))\nonumber\\
      &  & + \sum(f_{jm}-f^{\ast}_{jm})(\sum k^{\ast}_{ji_m}-l_m\mu^{\ast}_j)\nonumber\\
      &  & + \sum(\lambda_{j}-\lambda^{\ast}_{j})(\sum x^{\ast}_{ji_m} - B_j)\nonumber\\
      &  & + \sum(k_{ji_m}-k^{\ast}_{ji_m})(x^{\ast}_{ji_m}-f^{\ast}_{jm}r_m)\nonumber\\
      &  & + \sum(\mu_{j}-\mu^{\ast}_{j})(\sum f^{\ast}_{ji_m}V_m- S_j)\nonumber\\
      & \leq & \sum(x_{ji_m}-x^{\ast}_{ji_m})(g_{x_{ji_m}}-g_{x^{\ast}_{ji_m}})\nonumber\\
      &  & + 0 + 0 + 0 + 0 + 0 \leq 0\nonumber
\end{eqnarray}
where the first inequality is obtained by dropping the $^{[a,b]}_y$ terms, and the second inequality is obtained by applying the set of KKT conditions. The last set inequality holds due to the concavity of the function $\min{(\sum x_{ji_m},r_m})$ over $x_{ji_m}$.

Using (a) and (b), it follows from Krasovskii-LaSalle principle~\cite{lasalle1960some} that $\mathbf{y}$ converges to the set $\set{S}=\{\mathbf{y}|\dot{V}(\mathbf{y})=0\}$. It remains to show that the $\set{S}$ contains no trajectories other than $\{\mathbf{y}=\mathbf{y^{\ast}}\}$. Due to the space constraint, we omit that part of the proof in this paper. Interested readers can refer to~\cite{chen2008ump} for details.
\end{proof}

We make the following remarks:
\begin{itemize}
  \item The resource allocation algorithm has intuitive economic explanations. Specifically, one can view $k_{ji_m}$ and $\lambda_j$ as the \textit{video availability prices} and \textit{bandwidth prices} which are induced by helper $j$'s storage and bandwidth constraints and which helper $j$ \textit{charges} user $i_m$. One can also view $\mu_j$ as \textit{storage prices} that helper $j$ has to \textit{pay}. Indeed, the larger the video availability and bandwidth prices are, the smaller $\dot{x}_{ji_m}$ (that user $i_m$ requests from helper $j$) is. The larger/smaller the video availability price/storage price is, the larger $f_{jm}$ (that helper $j$ increases for video $m$) is. Similarly, the values of the prices are also driven by the relative difference between the given demand and the available resources. For example, the increase in video availability price $k_{ji_m}$ is proportional to the difference between the demand $x_{ji_m}$ and the available rate $f_{jm}r_m$. This economic framework can be potentially extended to building incentive mechanisms into the system.
  \item It is also not hard to see that the primal variables $\mathbf{x}$ and $\mathbf{f}$ will converge to the following intuitive solutions. In problem~(\ref{equa:primal}), every user $i_m$ will choose to request $x_{ji_m}$ with the smallest combined prices $(\lambda_j+k_{ji_m})$ until it reaches the maximum possible value. If the summation of received rates has not reached $r_m$, it will choose to request $x_{ji_m}$ with the second smallest combined prices. It will continue to do so until the summation of the received rates reaches $r_m$. Similarly, the solution for $f_{jm},m=1,2,\ldots,M$ can be obtained by water-filling helper $j$'s storage $S_j$ in descending order of the combined prices $(r_m\sum_{i_m\in \set{N}^c_{j,m}}k_{ji_m}-\mu_jV_m)$, which matches with helper $j$'s goal to maximize its ``profit''.
\end{itemize}

\subsection{Markov Approximation of Overlay Optimization}
Recall that $U(c)$ is the optimal solution to problem~(\ref{equa:fixc}). It is then left to solve:
\begin{eqnarray}\label{equa:solvec}
    \max_{c}U(c)\;\;\;\;\;\;\mbox{s.t.}\;\;\;c\in \set{C}
\end{eqnarray}
However, the set of possible overlay configurations given peers neighborhood constraints is exponential in the number of nodes, which make the problem NP hard even in a centralized manner. To overcome this difficulty, we re-write it as follows:
\begin{eqnarray}\label{equa:MWIS}
    \max_{p} &  & \sum_{c\in \set{C}}p_cU(c)\\
    \mbox{s.t.} &  & \sum_{c\in \set{C}}p_c = 1\;\;\;\mbox{and}\;\;\; 0\leq p_c\leq 1,\;\forall c\in \set{C}
    \nonumber
\end{eqnarray}
One can see that the problems~(\ref{equa:MWIS}) and~(\ref{equa:solvec}) are equivalent: the optimal solution to problem~(\ref{equa:MWIS}) is obtained by setting $p_{c^{*}}=1$ for $c^{*}=\argmax_{c'\in\set{C}}U(c')$ and $p_{c}=0$ for all other $c\in\set{C}$. Relaxing the objective $\sum_{c\in\set{C}}p_{c}U(c)$ by adding a weighted entropy term $\frac{1}{\kappa}H(p)$, where
$\kappa>0$ and $H(p)=-\sum_{c\in\set{C}}p_{c}\log{p_{c}}$, we have the following theorem shown by Chen et al.~\cite{chen2009markov}:
\begin{theorem}
The optimal solution to:
\begin{eqnarray}\label{equa:MWIS_approx}
    \max_{\mathbf{p}} &  & \sum_{c\in \set{C}}p_cU(c) - \frac{1}{\kappa}\sum_{c\in \set{C}}p_c\log{p_c}\\
    \mbox{s.t.} &  & \sum_{c\in \set{C}}p_c = 1\;\;\;\mbox{and}\;\;\; 0\leq p_c\leq 1,\;\forall c\in \set{C}\nonumber
\end{eqnarray}
is given by:
\begin{eqnarray}\label{equa:MC_solution}
    p_c^{\ast} = \frac{\exp{(\kappa U(c))}}{\sum_{c\in \set{C}}\exp{(\kappa U(c))}},\;\forall c\in \set{C}
\end{eqnarray}
\end{theorem}
\begin{proof}
The Lagrangian of the problem~(\ref{equa:MWIS_approx}) is given by:
\begin{eqnarray}
    L(p_c,\nu_c,\mu)& = &\sum_{c\in \set{C}}p_cU(c) - \frac{1}{\kappa}\sum_{c\in \set{C}}p_c\log{p_c}\nonumber\\
     &  & + \sum_{c\in \set{C}}\nu_cp_c + \mu(1 - \sum_{c\in \set{C}}p_c),
    \nonumber
\end{eqnarray}
where $\nu_c$ and $\mu$ are the Lagrangian variables. At optimal, the following KKT conditions~\cite{boyd2004convex} should hold:
\begin{eqnarray}
    U(c) - \frac{1}{\kappa}(\log{p^{\ast}_c}+1) + \nu^{\ast}_c - \mu^{\ast} = 0,\;\forall c\in \set{C},\nonumber
\end{eqnarray}
where $p^{\ast}_c$ is the primal optimal, and $\nu^{\ast}_c$ and $\mu^{\ast}$ are the dual optimal.
Writing $p^{\ast}$ as a function of $\nu^{\ast}_c$ and $\mu^{\ast}$ and applying the constraint $\sum_{c\in \set{C}}p^{\ast}_c=1$, we obtain:
\begin{eqnarray}
    \mu^{\ast} = \frac{1}{\kappa}\log{(\sum_{c\in \set{C}}\exp{(\kappa(U(c)+\nu^{\ast}_c)-1)})}.\nonumber
\end{eqnarray}
Plugging $\mu^{\ast}$ back into $p^{\ast}_c$, we get:
\begin{eqnarray}
    p^{\ast}_c & = & \frac{\exp{(\kappa(U(c)+\nu^{\ast}_c)-1)}}{\exp{(\kappa\mu^{\ast})}}\nonumber\\
     & = & \frac{\exp{(\kappa U(c))}}{\sum_{c\in \set{C}}\exp{(\kappa U(c))}},\;\forall c\in \set{C}.\nonumber
\end{eqnarray}
\end{proof}

Note that the optimal solution $p_{c}^{*}$ is in a product-form, and thus is the stationary distribution of some time-reversible Markov Chain (MC), hence the term MC approximation. Note that as $\kappa\rightarrow+\infty$, $p_{c^{*}}^{*}\rightarrow 1$ and therefore the optimal solution of the relaxed problem~(\ref{equa:MWIS_approx}) approaches to that of the original problem~(\ref{equa:solvec}). It is also easy to see that for a fixed $\kappa$, the error term $\frac{1}{\kappa}H(p)$ is bounded by $\frac{1}{\kappa}\log{|\set{C}|}$.

Our motivation behind this approximation is that it can potentially lead to distributed solutions. In this case, one can construct a MC with the overlay topology configurations as its states, and carefully design transition rates $q_{c,c'}$ such that the overall system will probabilistically jump between topology configurations while staying in the best configuration, i.e., $c^{*}$ for most of the time, and that the system performance will approach to the optimal. One straightforward design of $q_{c,c'}$ is given by:
\begin{eqnarray}\label{equa:idealMC}
q_{c,c'} = \left\{
              \begin{array}{ll}
               \frac{\tau}{\exp{\left(\kappa\left(U(c)\right)\right)}} & \hbox{$c,c'$ satisfy $S$;} \\
               0 & \hbox{otherwise.}
              \end{array}
           \right.
\end{eqnarray}
where $\tau>0$ is a constant, $U(c)$ is the overall system utility under state $c$, and $S$ is the following set of conditions:
\begin{itemize}
  \item $\exists\tilde{c}$ s.t. $\tilde{c}\subseteq c, \tilde{c}\subseteq c', |c\setminus\tilde{c}|=|c'\setminus\tilde{c}|=1$;
  \item Link $c\setminus\tilde{c}$ and link $c'\setminus\tilde{c}$ originates from the \textit{same} peer.
\end{itemize}
In other words, only the following state transitions $c\rightarrow c'$ are valid: \textit{a single peer} first drops a \textit{single} connection to one of his neighbors (from $c\rightarrow\tilde{c}$) and then randomly adds a new \textit{single} connection from his neighborhood (from $\tilde{c}\rightarrow c'$). $\tilde{c}$ is an auxiliary state and can be viewed as the intermediate state in which a single link from a single peer is dropped from $c$, where $c\setminus\tilde{c}$ represents the dropped link. It is not hard to see that $q_{c,c'}$ satisfies the detailed balanced equations $q_{c,c'}p_{c}^{*}=q_{c',c}p_{c'}^{*}$, thus the stationary distribution $p_{c}^{*}$ in equation~(\ref{equa:MC_solution}) can be achieved. We refer to this as the ``\textit{uniform-neighbor-choking}" algorithm because peers uniformly randomly choke neighbors in periods that depend on $U(c)$.

However, one caveat of the above design is that a peer still needs to know the global information $U(c)$ that needs to be broadcast to all the peers from time to time. This burdens the system with overhead. It is desirable to have a distributed algorithm in which each peer needs only \textit{local} information to perform such update and still achieve \textit{global} optimality.

\subsection{Soft-Worst-Neighbor-Choking Algorithm}
Motivated by the above discussions, we propose the ``\textit{soft-worst-neighbor-choking}" algorithm by designing:
\begin{eqnarray}\label{equa:modified_MC}
\bar{q}_{c,c'} = \frac{\tau}{\exp{(\kappa x^{c\setminus\tilde{c}})}}
\end{eqnarray}
where $x^{c\setminus\tilde{c}}$ is the rate on the dropped link $c\setminus\tilde{c}$, and $c,\tilde{c},c'$ should satisfy $S$. Here, the transition rates depend on only local information of link rates of peers' active neighbors.

We now give the overall distributed algorithm. For simplicity, the algorithm is stated under the perspective of user $i_m$, and those at other users/helpers are similar and self-explanatory.\\
\textbf{Topology building - ``soft-worst-neighbor-choking"}
\begin{itemize}
\item Initialization: User $i_m$ randomly chooses and connects to $N^{\max}_{i_m}$ neighboring helpers from his neighborhood $\set{N}_{i_m}$ and does the following steps.
\item Step 1: User $i_m$ independently draws an exponentially distributed random variable with mean\\ $\frac{1}{\tau(|\set{N}_{i_m}|-N^{\max}_{i_m})\sum_{j\in\set{N}^c_{i_m}}{\exp{(-\kappa x_{ji_m})}}}$ and counts down to zero.
\item Step 2: After the count-down expires, user $i_m$ drops neighbor $j$ with probability $\frac{\exp{(-\kappa x_{ji_m})}}{\sum_{j'\in\set{N}^c_{i_m}}{\exp{(-\kappa x_{j'i_m})}}}$ and randomly chooses and connects to a new neighbor from the set $\set{N}_{i_m}\setminus\set{N}^c_{i_m}$. It then repeats Step 1.
\end{itemize}
We make the following remarks:
\begin{itemize}
  \item It is easy to see that the algorithm gives $q_{c,c'}$ as in~(\ref{equa:modified_MC}).
  \item The algorithm is fully distributed, i.e., each peer runs the algorithm independently. Compared to the ``uniform-neighbor-choking" algorithm, peers only need to know local information of the link rates of their one-hop neighbors. The algorithm is intuitive: the larger the link rate, the less likely it is dropped and vice versa.
  \item It is worth noting that BitTorrent~\cite{cohen2003incentives} uses a ``\textit{worst-neighbor-choking}" algorithm, where each peer periodically chokes the link with the worst rate. In our case, link rates are weighted exponentially. The worst link is choked with the highest probability (which approaches $1$ as $\kappa\rightarrow+\infty$) while other links can also be choked occasionally, hence the term ``\textit{soft-worst-neighbor-choking}" algorithm. This algorithm is also generalizable to other P2P systems.
\end{itemize}

\subsection{Performance Analysis of Soft-Worst-Neighbor-Choking}
We now state the mathematical underpinnings behind such design of the algorithm and analyze its performance. It is interesting to see that the rate $x^{c\setminus\tilde{c}}$ on link $c\setminus\tilde{c}$ can be viewed as an approximation to $U(c)-U(\tilde{c})$, which is the overall system performance difference before and after link $c\setminus\tilde{c}$ is dropped. If $x^{c,\tilde{c}} = U(c)-U(\tilde{c})$, i.e., the helper cannot re-utilize the upload rate on $c\setminus\tilde{c}$ after it is dropped, then we have:
\begin{eqnarray}
\bar{q}_{c,c'} = \frac{\tau}{\exp{\left(\kappa\left(U(c)-U(\tilde{c})\right)\right)}}\nonumber
\end{eqnarray}
where $c,\tilde{c},c'$ satisfy $S$. In this case, $\bar{q}_{c,c'}$ still satisfies $\bar{q}_{c,c'}p_{c}^{*}=\bar{q}_{c',c}p_{c'}^{*}$ and the stationary distribution is no different from that of the uniform-neighbor-choking algorithm in~(\ref{equa:MC_solution}). However, $U(c)-U(\tilde{c}) \leq x^{c,\tilde{c}}$ in general, because the rate $x^{c,\tilde{c}}$ on the dropped link maybe fully or partially re-utilized by the helper for his other neighbors. In the following, we show that under some minor assumptions, one can still achieve a stationary distribution in product form similar to that in~(\ref{equa:MC_solution}).

Denote by $\omega_c = x^{c,\tilde{c}} - \left(U(c)-U(\tilde{c})\right)$ the error term resulted by approximating $\left(U(c)-U(\tilde{c})\right)$ with $x^{c\setminus\tilde{c}}$. Depending on overlay $c$ and the actual converged values of the storage and rate allocation algorithm, $\omega_c$ may take values anywhere in between $0$ and $B_{\max}$, where $B_{\max}$ is the maximum over all helpers' upload capacity. We quantize such error $\omega_c$ into $n_c+1$ values $[0,\frac{B_{\max}}{n_c},\frac{2B_{\max}}{n_c},\ldots,B_{\max}]$, and assume that $\omega_c=\frac{kB_{\max}}{n_c}$ with probability $\rho_{c_{k}},k=0,1,\ldots,n_{c}$ and $\sum_{k=0}^{n_{c}}{\rho_{c_{k}}}=1$. Under these assumptions, we show the following theorem:
\begin{theorem}
The stationary distribution $p_c$ of MC with transition rates~(\ref{equa:modified_MC}) is given by:
\begin{eqnarray}\label{equa:pi_conglomerate_modifiedMC}
p_{c} = \sum^{n_c}_{k=0}{p_{c_{k}}} =  \frac{\sigma_c\exp{\left(\kappa U(c)\right)}}{\sum_{c'\in\set{C}}{\sigma_{c'}\exp{\left(\kappa U(c')\right)}}}
\end{eqnarray}
where $\sigma_c=\sum_{k=0}^{n_c}{\rho_{c_k}\exp{\left(\kappa\frac{kB_{\max}}{n_c}\right)}}$.
\end{theorem}
\begin{proof}
Consider a modified MC as follows: expand each state $c$ of the original MC to $n_c+1$ states $c_{k},k=0,1,\ldots,n_c$ with the following transition rates:
\begin{eqnarray}\label{equa:expanded_MC}
\bar{q}_{c_k,c'_{k'}} = \frac{\tau\rho_{c'_{k'}}}{\exp{\left(\kappa\left(U(c)-U(\tilde{c})+\frac{kB_{\max}}{n_c}\right)\right)}}
\end{eqnarray}
where $\rho_{c'_{k'}},k=0,1,\ldots,n_{c'}$, is the probability measure on expanded states and $\sum_{k'=0}^{n_{c'}}{\rho_{c'_{k'}}}=1$. Note that $c_0$ refers to state $c$ with zero error. Using equation~(\ref{equa:expanded_MC}) and detailed balance equations $p_{c_k}\bar{q}_{c_k,c'_{k'}} = p_{c'_{k'}}\bar{q}_{c'_{k'},c_k}$, we have $\forall c_0,c'_{k'}$:
\begin{eqnarray}\label{equa:db_result}
\frac{p_{c_0}}{\rho_{c_{0}}\exp{\left(\kappa U(c)\right)}} = \frac{p_{c'_{k'}}}{\rho_{c'_{k'}}\exp{\left(\kappa\left(U(c')+\frac{k'B_{\max}}{n_{c'}}\right)\right)}}=\hbox{\textit{const}}\nonumber
\end{eqnarray}
Using $\sum_{c\in\set{C}}{\sum^{n_c}_{k=0}{p_{c_k}}}=1$, we obtain:
\begin{eqnarray}\label{equa:pi_modifiedMC}
p_{c_{k}} = \frac{\rho_{c_k}\exp{\left(\kappa\left(U(c)+\frac{kB_{\max}}{n_c}\right)\right)}}{\sum_{c'\in\set{C}}{\sum_{k'=0}^{n_{c'}}{\rho_{c'_{k'}}\exp{\left(\kappa\left(U(c')+\frac{k'B_{\max}}{n_{c'}}\right)\right)}}}}\nonumber
\end{eqnarray}
Denote by $\sigma_c=\sum_{k=0}^{n_c}{\rho_{c_k}\exp{\left(\kappa\frac{kB_{\max}}{n_c}\right)}}$ and we have:
\begin{eqnarray}
p_{c} = \sum^{n_c}_{k=0}{p_{c_{k}}} = \frac{\sigma_c\exp{\left(\kappa U(c)\right)}}{\sum_{c'\in\set{C}}{\sigma_{c'}\exp{\left(\kappa U(c')\right)}}}
\end{eqnarray}
\end{proof}
We make the following remarks:
\begin{itemize}
  \item If $n_c=const\;\forall c$ and the distribution $\rho_{c_k},k=0,1,\ldots,n_c$ is the same $\forall c$, $\sigma_c = const$. In this case, $p_c=p^{*}_c$.
  \item The total variational distance $d_{TV}(\mathbf{p}^{*},\mathbf{p})$ can be upper bounded by $(1-\exp{(-\kappa B_{\max})})$. This is because $p^{*}_c-p_c=p^{*}_c\left(1-\frac{\sigma_c}{\frac{\sum_{c'\in\set{C}}{\sigma_{c'}\exp{(\kappa(U(c')))}}}{\sum_{c'\in\set{C}}{\exp{(\kappa(U(c')))}}}}\right)$. Since $\sigma_c\in[1,\exp{(\kappa B_{\max})}]$, the fractional can be lower bounded by $\exp{(-\kappa B_{\max})}$ and hence the result.
  \item In general, it is difficult to give a tight analytical lower bound on the performance $\sum_{c\in\set{C}}{U(c)p_{c}}$ because $\sigma_c$ is unknown. However, note that when $\kappa\rightarrow+\infty$, $p_{c^{*}}\rightarrow1$ where $c^{*}=\argmax_{c'\in\set{C}}{U(c')}$, and the system approaches the optimal $U(c^{*})$. We will show in our simulations that the algorithm performs well and is close to the optimal.
\end{itemize}

\subsection{Discussions}
Our overall scheme is fully distributed, i.e., each peer runs the algorithm independently and makes changes based on only one-hop local information. The user passes the derivative of their utility function to helpers to perform distributed resource allocation; users and helpers periodically choke their neighbors based the relative performance of their one-hop links.

The deployment of such a practical system is easy: newly deployed helper nodes can automatically connect to a set of interested users and load balance their storage and bandwidth resources using the distributed algorithms. This simple solution helps achieve minimum maintenance overhead and can easily adapt to system dynamics. When system's supply and demand pattern changes due to helper/user joining/leaving the system and video popularity shifting, the helper nodes will automatically update their content caching, allocate their bandwidth resource and dynamically change their neighborhood selections in a distributed and local manner, which will best match to the global system demand and available resources and optimize overall system utility. We summarize the simple algorithms at both classes of peer nodes in the next section.

\section{System Implementation}
\subsection{Back-up Server and Tracker}
A centralized server with all the video content is present that acts as the ``life-line'' to supplement the deficit (if any) in the system. A tracker is used to keep track of all the participating peers and to assist building an overlay network. When a user/helper joins the system, it obtains from the tracker the IP addresses of a list of helpers/users in their neighborhood. It then connects to its maximum allowed number of neighbors randomly chosen from their neighborhood.
\begin{algorithm}[htbp]
  \caption{User Protocol}
  \label{alg:user}
  \begin{algorithmic}[1]
   \STATE \textbf{Initialization: }Set $t_1,t_2=0$, draw $t_3\sim Exp\left(\tau(|\set{N}_{i_m}|-N^{\max}_{i_m})\sum_{j\in\set{N}^c_{i_m}}{\exp{(-\kappa x_{ji_m})}}\right)$, and iterate:
   \IF{mod$(t_1,T_{i_m})=0$}
       \STATE Count the number of received packets from each neighboring helper over the period $[t_1-T_{i_m},t_1-1]$, and update the corresponding average rate $x_{ji_m}$. Derive the derivative $g_{x_{ji_m}}$ of its utility function and sends them to all $j\in \set{N}^c_{i_m}$.
   \ENDIF
   \IF{mod$(t_2,\hbox{BUFFER\_TIME})=0$}
       \STATE Download from the server all the missing packets in the next BUFFER\_TIME worth of segments.
   \ENDIF
   \IF{$t_3=0$}
       \STATE Drop neighbor helper $j$ with probability $\frac{\exp{(-\kappa x_{ji_m})}}{\sum_{j'\in\set{N}^c_{i_m}}{\exp{(-\kappa x_{j'i_m})}}}$. Randomly choose and connect to a new neighbor from the remaining neighborhood to replace $j$, and set $x_{ji_m}=0$. Draw $t_3\sim Exp\left(\tau(|\set{N}_{i_m}|-N^{\max}_{i_m})\sum_{j\in\set{N}^c_{i_m}}{\exp{(-\kappa x_{ji_m})}}\right)$.
   \ENDIF
   \STATE $t_1 \leftarrow t_1 + 1, t_2\leftarrow t_2 + 1, t_3 \leftarrow t_3 - 1$.
  \end{algorithmic}
\end{algorithm}

\subsection{Packet Exchange Protocol}
The video packets each helper stores are coded using a rateless code and downloaded from the server. To the validate the proposed algorithms under asynchronous scenarios, we let each helper $j$ maintain its own clock. Helper $j$ also updates its bandwidth/storage allocation algorithm only in periods of $T_j$ seconds and keeps an outgoing buffer worth of $T_j$ seconds of its upload bandwidth capacity. Users maintain a buffer length denoted by BUFFER\_TIME that covers an integer number of video segments. For a particular user, as it decodes the packets and plays the video in the segment right ahead of its playback time, it also receives packets of the next unfulfilled segment from the helpers. Upon finishing BUFFER\_TIME worth of segments, the user will immediately fetch the packets from the server to fill any missing packets in the next BUFFER\_TIME range. Each user also has its own clock and an bandwidth request update period $T_{i_m}$. The detailed packet exchange protocols for both users and helpers are described in Algorithm~$1$ and $2$ based on the theoretical analysis given in section~(\ref{sec:distributed}).
\begin{algorithm}[htbp]
  \caption{Helper Protocol}
  \label{alg:helper}
  \begin{algorithmic}[1]
   \STATE \textbf{Initialization: }Set $x_{ji_m}=0$, $f_{jm}=0$, $\lambda_j=0$, $\mu_j=0$ and $k_{ji_m}=0$. Set $t_1=0$, draw $t_2\sim Exp\left(\tau(|\set{N}_{j}|-N^{\max}_{j})\sum_{i_m\in\set{N}^c_{j,m}}{\exp{(-\kappa x_{ji_m})}}\right)$, and iterate:
   \IF{mod$(t_1,T_{j})=0$}
       \STATE $x_{ji_m} \leftarrow x_{ji_m} + \alpha(g_{x_{ji_m}}-(\lambda_j+k_{ji_m}))_{x_{j_im}}^{[0,+\infty)}$
       \STATE $f_{jm} \leftarrow f_{jm} + \beta(\sum_{i_m\in \set{N}_j^m}k_{ji_m}-l_m\mu_j)_{f_{ji_m}}^{[0,1]}$
       \STATE $\lambda_j \leftarrow \lambda_j + \gamma(\sum_{m=1}^{M}\sum_{i_m\in \set{N}^c_{j,m}}x_{ji_m} - B_j)_{\lambda_j}^{[0,+\infty)}$
       \STATE $\mu_j \leftarrow \mu_j + \delta(\sum_{m=1}^{M}f_{ji_m}V_m - S_j)_{\mu_j}^{[0,+\infty)}$
       \STATE $k_{ji_m} \leftarrow k_{ji_m} + \varepsilon(x_{ji_m}-f_{jm}r_m)_{k_{ji_m}}^{[0,+\infty)}$.
       \STATE Allocate number of packets equivalent to $x_{ji_m}T_j$ to user $i_m$ for all $i_m$ and put them in the outgoing buffer.
       \STATE Re-allocate its storage of videos according to $f_{ji_m}$.
   \ENDIF
   \STATE Send remaining packets in the buffer to neighbor users.
   \IF{$t_2=0$}
       \STATE Drop neighbor user $i_m$ with probability $\frac{\exp{(-\kappa x_{ji_m})}}{\sum_{i'_m\in\set{N}^c_{j,m}}{\exp{(-\kappa x_{ji'_m})}}}$. Randomly choose and connect to a new neighbor from the remaining neighborhood to replace $i_m$, and set $x_{ji_m}=0$. Draw $t_2\sim Exp\left(\tau(|\set{N}_{j}|-N^{\max}_{j})\sum_{i_m\in\set{N}^c_{j,m}}{\exp{(-\kappa x_{ji_m})}}\right)$.
   \ENDIF
   \STATE $t_1 \leftarrow t_1 + 1, t_2 \leftarrow t_2 - 1$.
  \end{algorithmic}
\end{algorithm}

\section{Simulation Results}
It is worth noting that our proposed analysis relies on a few assumptions made to make mathematical arguments simple. First, we  have assumed in the algorithms that peers have \textit{synchronized} clocks. This can be difficult to maintain in practice. Second, when solving problem~(\ref{equa:solvec}), we have assumed that the underlying resource allocation algorithm has fully converged and ignored the different time scales. Although there exist a number of techniques that can address these issues~\cite{benveniste1990adaptive,chen2008ump}, they are not the focus of our paper and we omit the heavy discussions involved in the analysis. Instead, we validate the feasibility and effectiveness of the proposed algorithms by designing our simulations that capture real-world scenarios including the effects of asynchronized clocks among peers and random network delays. We also show in the simulations that the distributed scheme adapts well to system fluctuations including change in video demand patterns and peer dynamics.
\subsection{Experimental Setup}
We set total number of videos $M=4$, helpers $J=70$ and users $\sum_{m=1}^MI_m=100$. Table~\ref{tab:dist1} shows each video's streaming rate and the fraction of users watching it. Helpers have upload and storage capacities with different distributions shown in Tables~\ref{tab:dist2} and~\ref{tab:dist3}. Each peer can potentially connect to every other peer in the system, but has a maximum allowed number of neighbors uniformly randomly chosen from $[3,10]$. This setup is based on practical data in commercialized P2P systems~\cite{pplive,ppstream,xunlei}, which makes it easy to test the robustness of the proposed algorithms in highly heterogenous scenarios. We also set the step sizes $\alpha=1, \beta=0.01, \gamma=\delta=0.5, \varepsilon=0.05$ for the bandwidth and storage allocation algorithm, and set $\kappa = 10, \tau=0.01$ in the topology update algorithm. These parameters are chosen to guarantee smooth algorithm updates and small MC approximation errors.

\subsection{Convergence in the Static Case}
\label{sec:static}
We first test the convergence of the storage and bandwidth allocation algorithm in the static case, where all peers stay in the system during the entire simulation time and perform no topology update. We first focus on the synchronous case, where peers share a synchronous clock and have an update period of $1$ second. Figure~\ref{fig:static}(a) shows the instantaneous server load versus simulation time. Also shown as for comparison is the system's intrinsic deficit, i.e., total users' streaming rate demand minus total helpers' upload bandwidth. The initial server load is high, but it quickly drops to a stable point. The sub-figures (b) and (c) in Figure~\ref{fig:static} show the convergence of a particular helper's (ID $= 1$) upload rate and storage allocation. The convergence results for the shadow prices are similar, which we omit here due to limit in space.
\begin{table}[htbp]
\caption{Video streaming rate distribution}
\label{tab:dist1}
\vspace{-0.1in}
\begin{center}
\tabcolsep 5.8pt
\small
\begin{tabular}{|c|c|c|c|c|}
  \hline
  \hbox{Streaming rate (kbps)} & 768 & 896 & 896 & 1152 \\
  \hline
  \hbox{Fraction ($\%$)} & 10 & 20 & 50 & 20 \\
  \hline
\end{tabular}
\end{center}
\end{table}
\vspace{-0.25in}
\begin{table}[htbp]
\caption{Helper upload capacity distribution}
\label{tab:dist2}
\vspace{-0.1in}
\begin{center}
\tabcolsep 5.3pt
\small
\begin{tabular}{|c|c|c|c|c|c|c|c|}
  \hline
  \hbox{Upload (kbps)} & 256 & 384 & 512 & 640 & 768 & 896 & 1024 \\
  \hline
  \hbox{Fraction ($\%$)} & 5 & 10 & 15 & 40 & 15 & 10 & 5 \\
  \hline
\end{tabular}
\end{center}
\end{table}
\vspace{-0.25in}
\begin{table}[htbp]
\caption{Helper storage capacity distribution}
\label{tab:dist3}
\vspace{-0.1in}
\begin{center}
\tabcolsep 4.3pt
\small
\begin{tabular}{|c|c|c|c|c|c|c|c|}
  \hline
  \hbox{Storage (MB)} & 768 & 960 & 1152 & 1344 & 1536 & 1728 & 1920 \\
  \hline
  \hbox{Fraction ($\%$)} & 5 & 5 & 10 & 10 & 20 & 40 & 10 \\
  \hline
\end{tabular}
\end{center}
\end{table}
\vspace{-0.1in}
\begin{figure*}
\centering
\subfigure[Server load, sync.]{
\includegraphics[height=1.6in, width=2.2in]{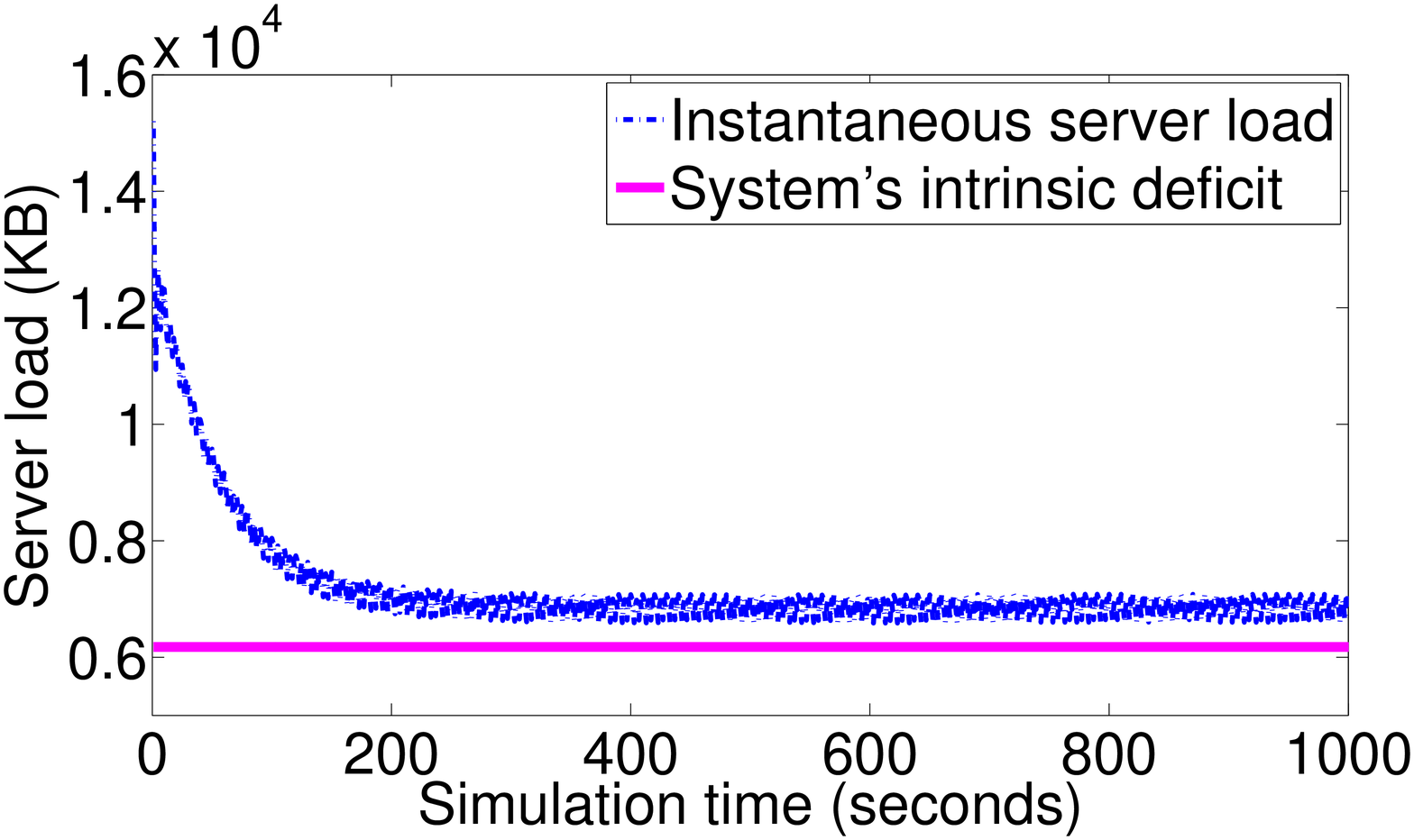}
}
\hspace{-0.6cm}
\subfigure[Bandwidth alloc., sync.]{
\includegraphics[height=1.6in, width=2.2in]{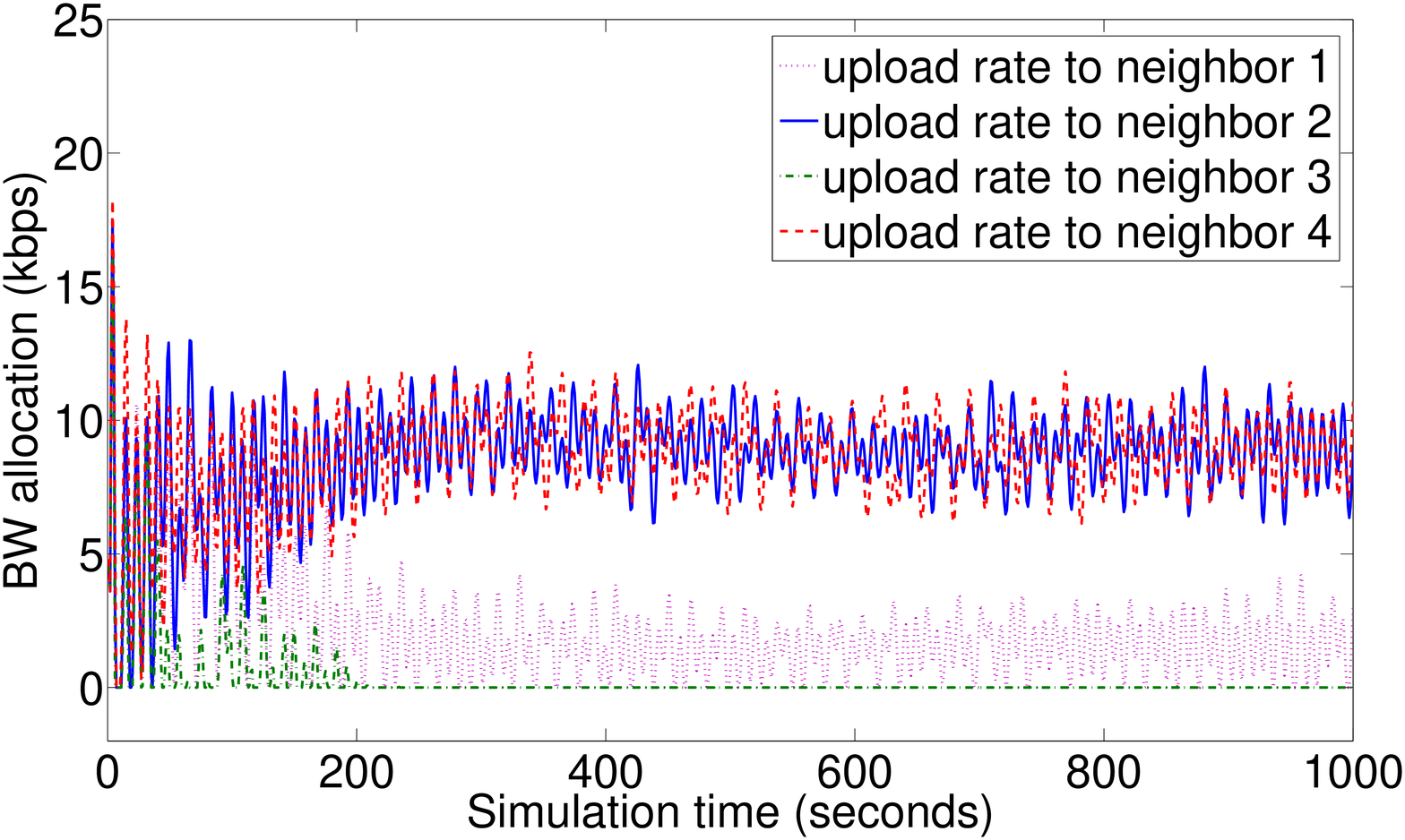}
}
\hspace{-0.6cm}
\subfigure[Storage alloc., sync.]{
\includegraphics[height=1.6in, width=2.2in]{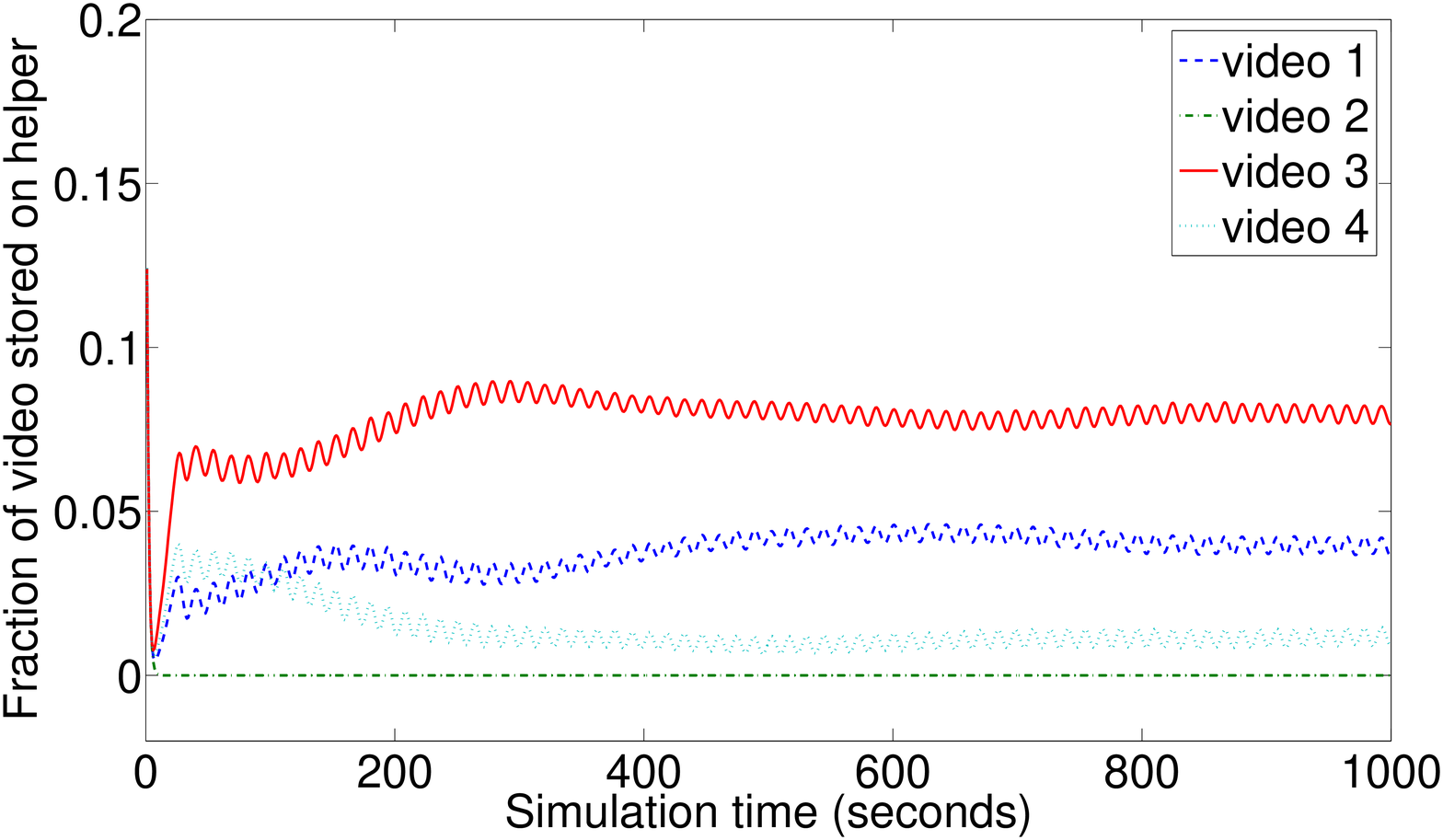}
}
\caption{Convergence results of storage and bandwidth allocation algorithm in a static and synchronous setting, where no overlay topology update nor peer dynamics is present. Peers have synchronized update periods of $1$ seconds. (a), (b) and (c) show the required server load, bandwidth and storage allocation for the helper (ID $=1$).}
\label{fig:static}
\end{figure*}

\subsection{Asynchrony and Random Network Delay}
To test the robustness of the bandwidth and storage allocation algorithm in real networks, we add asynchrony and random network delay in the system. Specifically, peers have asynchronous clocks and choose update periods uniformly randomly from the set of $\{1,3,5,7,9\}$ seconds. In addition, every peer has a communication delay to every one of its neighbors randomly chosen from $1$ to $5$ seconds. These numbers are chosen to stress test the system. Figure~\ref{fig:delay} shows the server load, bandwidth and storage allocation for the same helper (ID $=1$). Compared to Figure~\ref{fig:static}, the bandwidth allocations experiences more fluctuations, but they still center around comparable average values. The server load and helpers' storage load are quite stable, which demonstrates the robustness of the algorithm. In the following sections, our simulation experiments will apply the same asynchrony and random network delays unless mentioned otherwise.
\begin{figure*}
\centering
\subfigure[Server load, async.]{
\includegraphics[height=1.6in, width=2.2in]{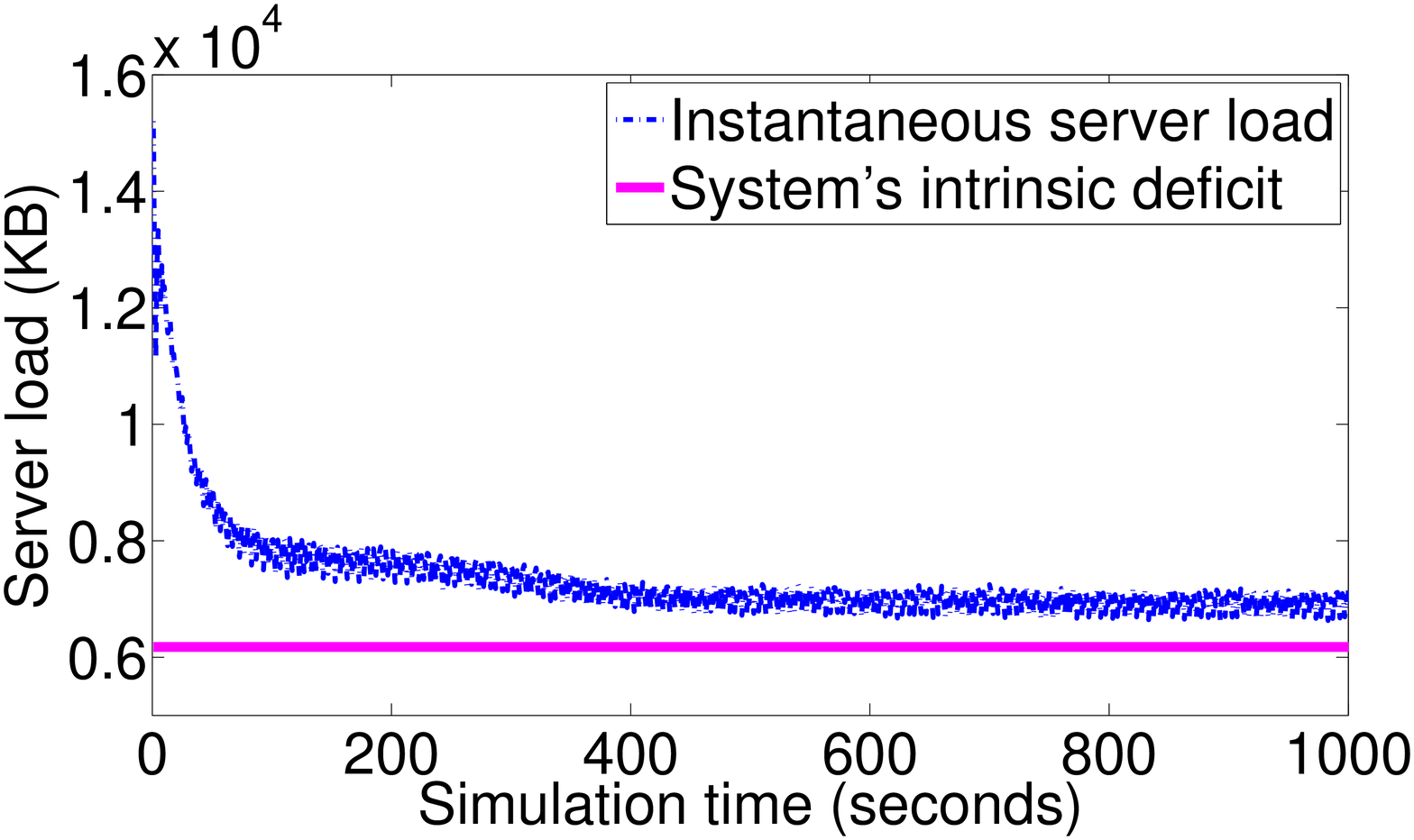}
}
\hspace{-0.6cm}
\subfigure[Bandwidth alloc., async.]{
\includegraphics[height=1.6in, width=2.2in]{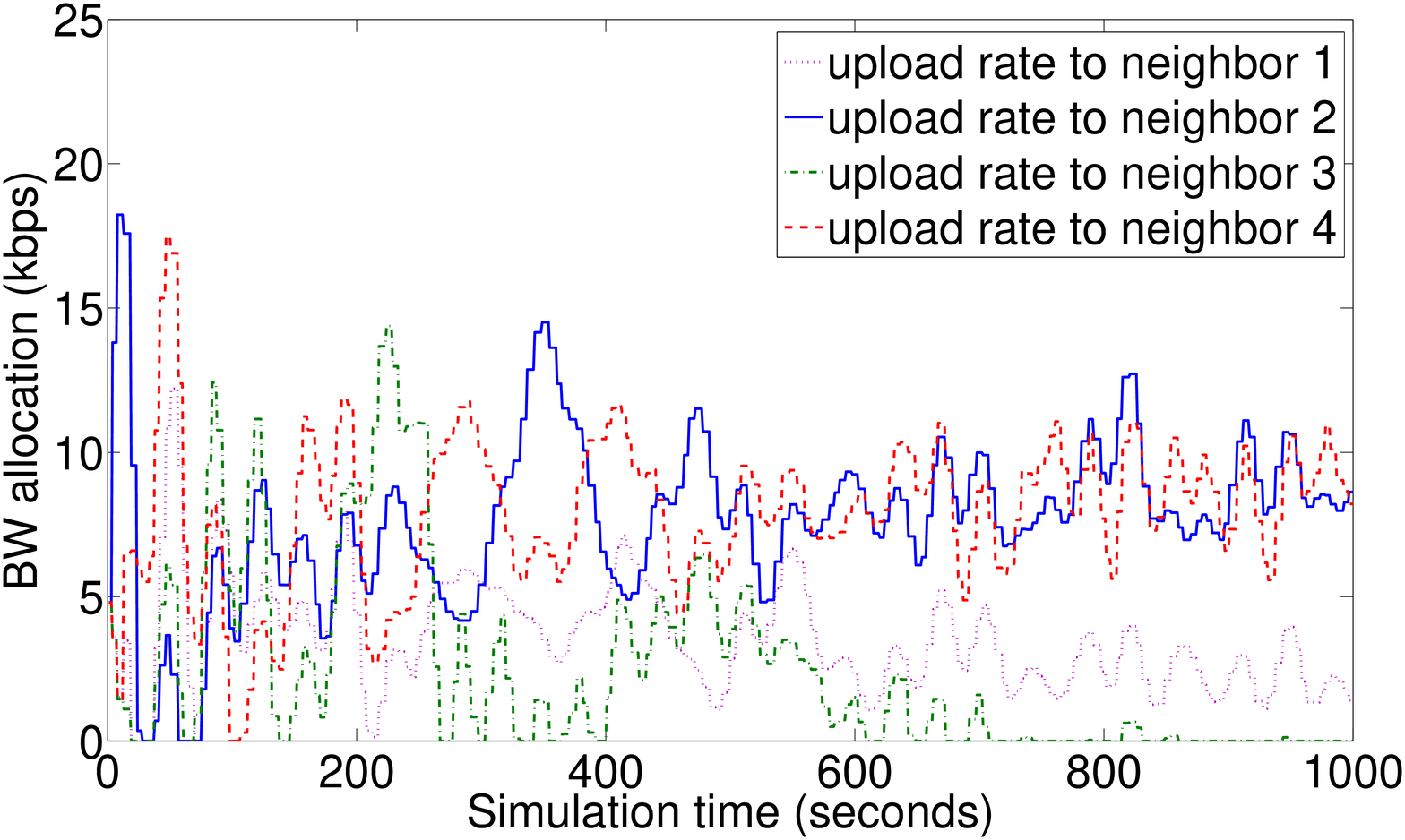}
}
\hspace{-0.6cm}
\subfigure[Storage alloc., async.]{
\includegraphics[height=1.6in, width=2.2in]{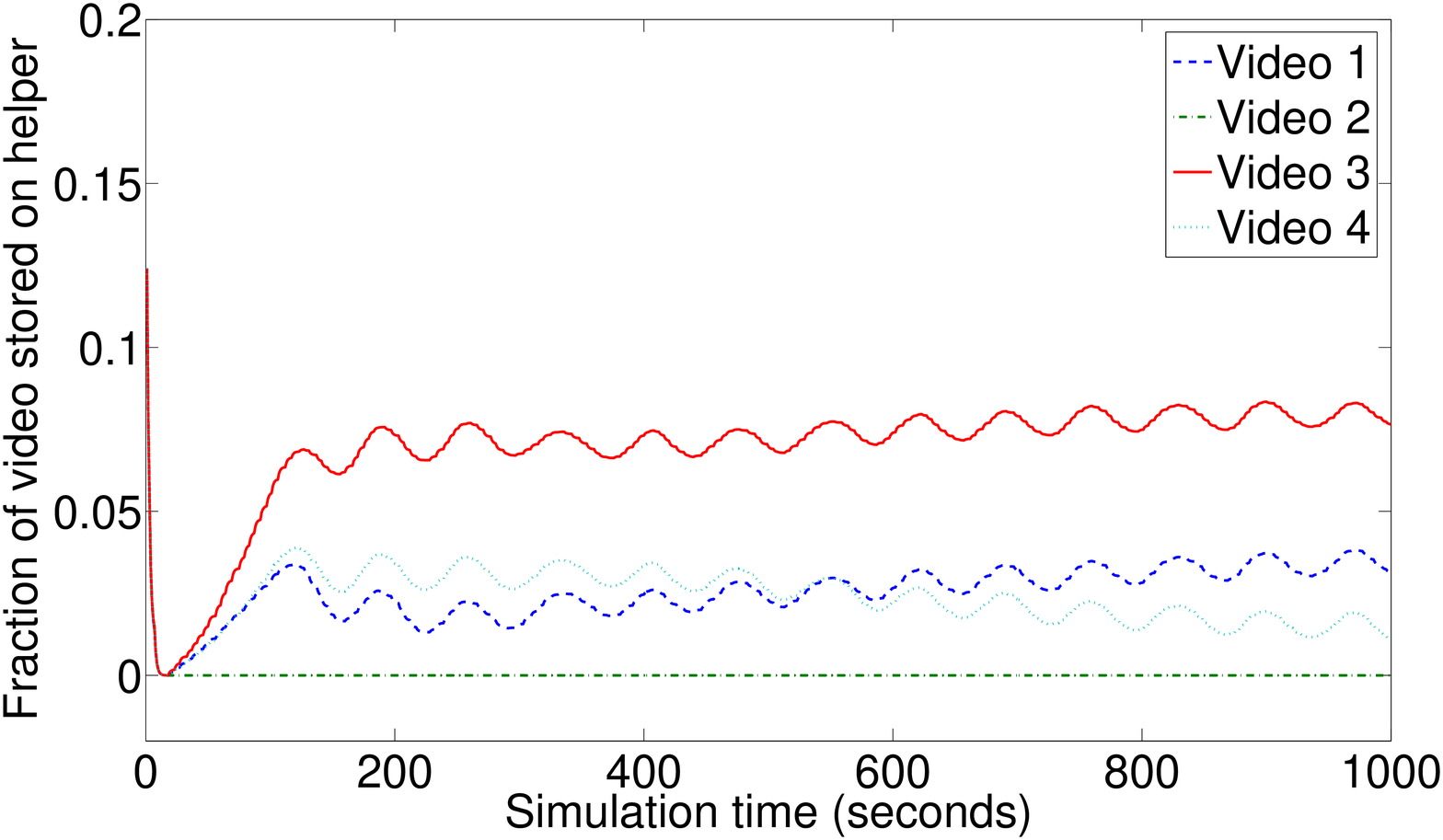}
}
\caption{Convergence of storage and bandwidth allocation algorithm in the case of asynchrony and random network delays, with no overlay topology update nor peer dynamics. Peers have asynchronous clocks and random update periods uniformly drawn from $\{1,3,5,7,9\}$ seconds. Each peer also has a communication delay to his neighbors randomly chosen from $[1,5]$. (a), (b) and (c) show the required server load, bandwidth and storage allocation for the helper (ID$=1$).}
\label{fig:delay}
\end{figure*}

\subsection{Effectiveness of the Overlay Topology Update}
As is evident from Figure~\ref{fig:static}, the server load cannot reach to the minimum value of the intrinsic system deficit without topology update. This is because some helpers have poor performance to their connected neighbors and have not fully utilized their upload bandwidth. We run the distributed overlay topology update algorithm on top of the resource allocation algorithm with other parameters and configurations unchanged. Figure~\ref{fig:change}(a)(b) shows server load versus simulation time without and with overlay topology update respectively, where Figure~\ref{fig:change}(a) is simply Figure~\ref{fig:delay}(a) shown again for comparison. It can be seen that overlay topology update buys approximately $14\%$ reduction in server load and eventually achieves the intrinsic system deficit which is the theoretical lower bound of server load.
\begin{figure}[htbp]
\centering
\subfigure[]{
\includegraphics[height=1.4in, width=1.6in]{fig1.eps}
}
\hspace{-0.6cm}
\subfigure[]{
\includegraphics[height=1.4in, width=1.6in]{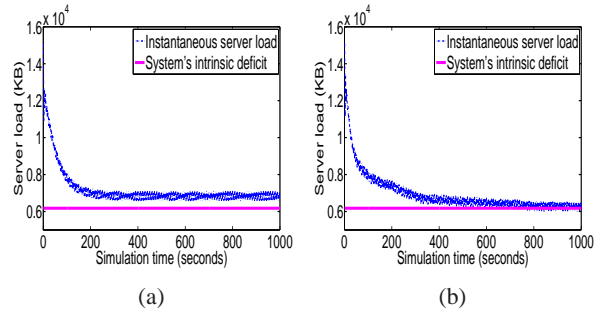}
}
\caption{Server load versus simulation time. Peers have asynchronous update periods and random network delays. (a) Server load without overlay topology update; (b) server load with overlay topology update.}
\label{fig:change}
\end{figure}

\subsection{Effects of Dynamics}
We show in this section that our system is ``plug-and-play", i.e., which requires minimum maintenance and which is automatically adaptive to system dynamics. Peers will only need to run their distributed algorithms regardless of system fluctuations and be able to keep updated to the supply and demand patterns across multiple channels.

We first examine the effects of peer dynamics. To do this, we add new users and new helpers that join the system following a Poisson process with mean $20$. The newly joined peers will follow the demand and resource distributions listed in Tables~\ref{tab:dist1},~\ref{tab:dist2}~and~\ref{tab:dist3}. In addition, every peer will stay in the system for an exponential random amount of time with average of $200$ seconds. To examine how fast the system responds to dynamics, we simulate till $1000$ seconds but stop the dynamic process at the $600$th second. Figure~\ref{fig:dynamic} (a)(b) show how the server load varies with time, without and with overlay topology updates respectively. The available system resources also change due to dynamics, as is evident from the varying intrinsic system deficit shown in the figures. It is demonstrated that the algorithm can keep updated to the dynamics. When overlay topology update is present, the system can also approach the minimum server load. Note that the instantaneous intrinsic system deficit stops at a different value in two cases, only due to the difference in the pseudo-randomness generated by the computer with and without the topology update. The results have demonstrated the robustness of the resource allocation and topology update algorithms to system dynamics.

We also use a simple example to illustrate how the system responds to changes in video demand patterns. In particular, we pick a helper (ID $=4$) who has $10$ neighbor users with $7$ users watching video $3$ and $3$ users watching video $4$. At $t=300$, we let all the users in video $3$ ``switch channels" with half of them switching to video $4$ and the other half switching to video $2$. Figure~\ref{fig:fign} shows how the helper (ID $=4$) responds to such change by re-allocating its storage resources. Both Figure~\ref{fig:fign} and Figure~\ref{fig:static}(c) demonstrate the helper's ``plug-and-play" feature, i.e., helpers can automatically load balance its resources given system demand patterns.
\begin{figure}[htbp]
\centering
\subfigure[]{
\includegraphics[height=1.4in, width=1.6in]{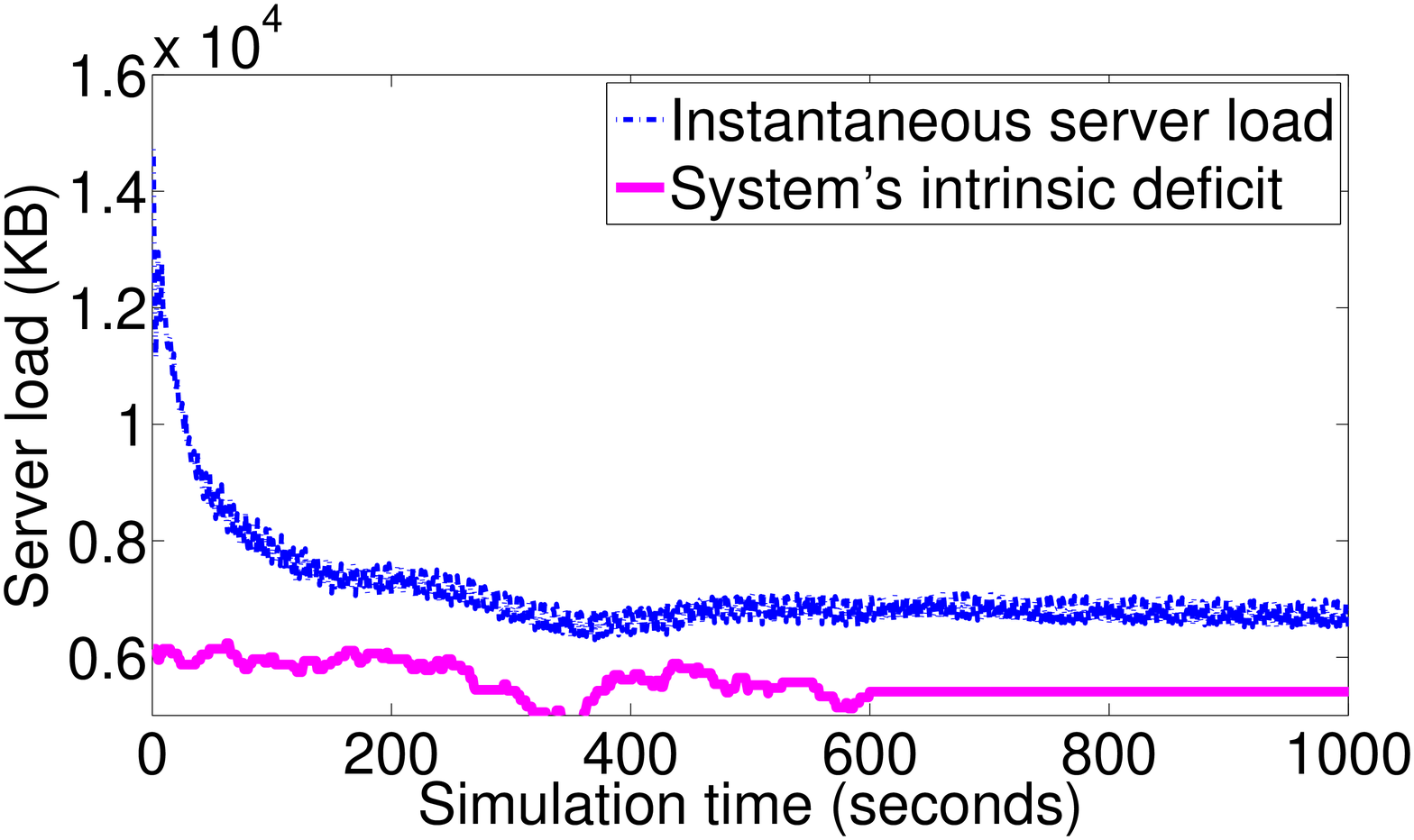}
}
\hspace{-0.6cm}
\subfigure[]{
\includegraphics[height=1.4in, width=1.6in]{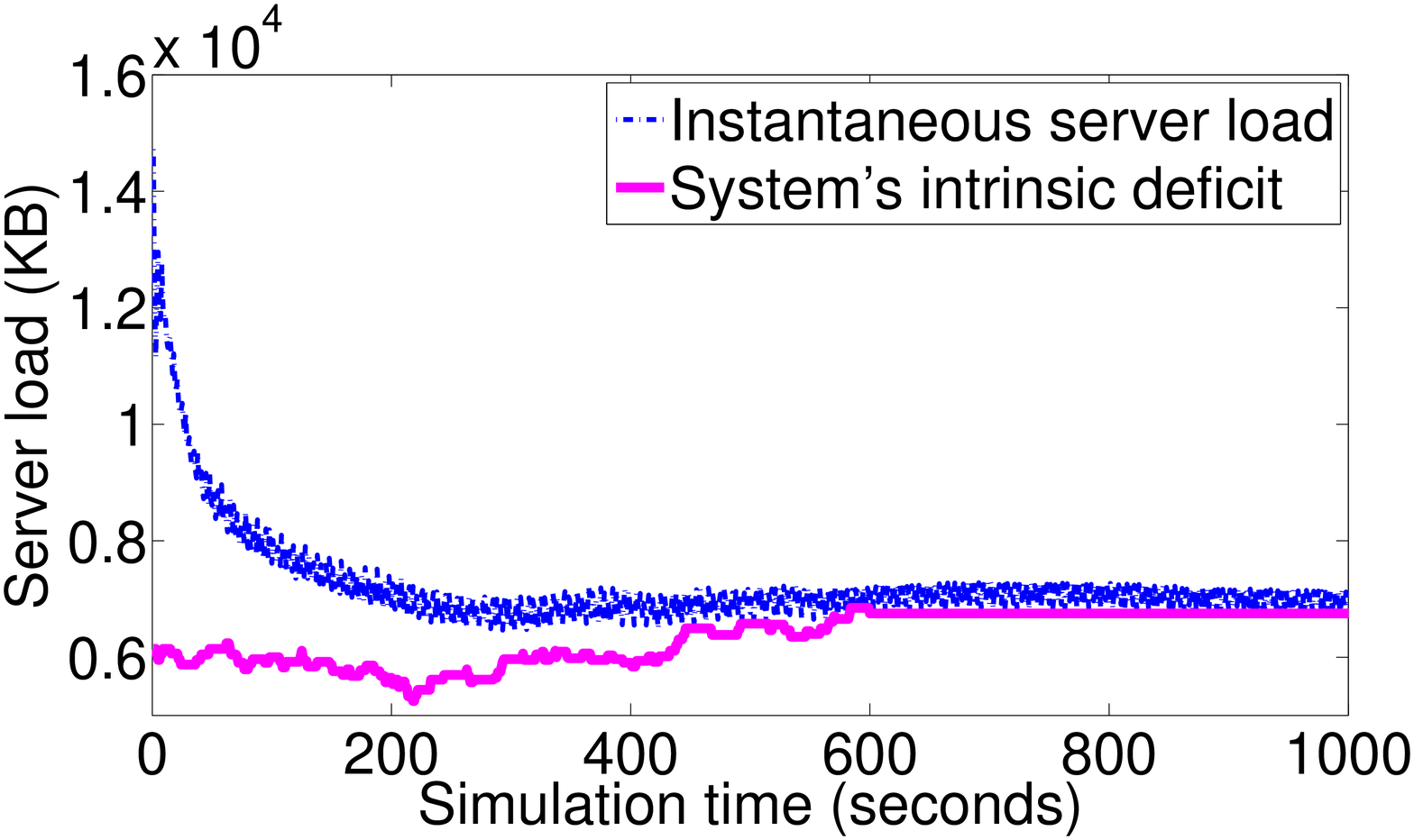}
}
\caption{Effects of system dynamics on server load. A new user and a new helper will join the system every $20$ seconds on average. Each peer stays for an average of $200$ seconds in the system. (a) Server load without overlay topology update; (b) server load with overlay topology update.}
\label{fig:dynamic}
\end{figure}
\begin{figure}[htbp]
\centering
\includegraphics[height=1.2in, width=1.6in]{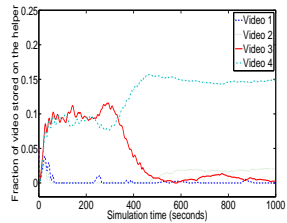}
\caption{Effects of changes in video demand patterns. Users in video session $3$ ``channel switch" to video $2$ and $4$ with equal numbers.}
\label{fig:fign}
\end{figure}

\section{Conclusions}
In this paper, we propose to minimize the server load in a helper-assisted multi-channel P2P VoD system. Helpers who help provide the VoD service are limited in bandwidth and storage, and each helper and user has a constraint on the maximum number of neighbors that they can connect to. This problem is critical for exploring the maximum potential of practical distributed P2P VoD systems. The mix-convex-combinatorial nature of the problem under practical constraints makes it challenging to solve even in a centralized manner. We tackle this challenge by designing two distributed algorithms running in tandem: a primal-dual resource allocation algorithm and a ``soft-worst-neighbor-choking'' topology building algorithm. The overall scheme is simple to implement and provably converges to a near-optimal solution. Simulation results show that our proposed algorithm minimizes the server load, and is robustness to asynchronous clock times, random network delay, video popularity changes and peers dynamics. Our proposed system design and algorithm provide useful insight to practical video content distribution applications. Possible future work includes: (1) design incentive mechanisms into the system; and (2) build a practical system prototype.

\bibliographystyle{IEEEtran}

\end{document}